\documentclass[a4paper]{amsart}
\usepackage{times}
\usepackage[latin1]{inputenc}
\usepackage{amsthm,amssymb}

\newtheorem{theorem}{Theorem}
\newtheorem{lemma}{Lemma}

\newcommand{\SetS}{\ensuremath{\mathcal{S}}}
\newcommand{\Q}{\ensuremath{Q}}
\newcommand{\E}{\ensuremath{\text{E}}}

\begin{document}
\title{A note on the generalized min-sum set cover problem}

\author{Martin Skutella}
\address{Institut f\"ur Mathematik, Sekr.~MA 5-2, Technische Universit\"at Berlin, Stra\ss e des 17.~Juni 136, 10623 Berlin, Germany.}
\email{martin.skutella@tu-berlin.de}
\author{David P.~Williamson}
\address{School of Operations Research and Information Engineering, Cornell
University, Ithaca, NY 14853, USA.}
\email{dpw@cs.cornell.edu}

\date{\today}

\maketitle

\begin{abstract}
In this paper, we consider the generalized min-sum set cover problem, introduced by Azar, Gamzu, and Yin \cite{AzarGY09}.  Bansal, Gupta, and Krishnaswamy \cite{BansalGK10} give a 485-approximation algorithm for the problem.  We are able to alter their algorithm and analysis to obtain a 28-approximation algorithm, improving the performance guarantee by an order of magnitude.  We use concepts from $\alpha$-point scheduling to obtain our improvements.
\end{abstract}

\section{Introduction}

In this note, we consider the generalized min-sum set cover problem.  In this problem we are given as input a universe $U$ of $n$ elements, a collection $\SetS = \{S_1,\ldots, S_m\}$ of subsets $S_i$ of $U$, and a covering requirement $K(S)$ for each $S \in \SetS$, where $K(S) \in \{1,2,\ldots,|S|\}$.  The output of any algorithm for the problem is an ordering of the $n$ elements.  Let $C_{S}$ be the position of the $K(S)$th element of $S$ in the ordering.  The goal is to find an ordering that minimizes $\sum_{S \in \SetS} C_{S}$.  This problem is a generalization of the min-sum set cover problem (in which $K(S) = 1$ for all $S \in \SetS$), introduced by Feige, Lov\'asz, and Tetali \cite{FeigeLT04}, and the min-latency set cover problem (in which $K(S)=|S|$ for all $S \in \SetS$), introduced by Hassin and Levin \cite{HassinL05}.  This generalization was introduced by Azar, Gamzu, and Yin \cite{AzarGY09} in the context of a ranking problem.

Because the problem is NP-hard, Azar, Gamzu, and Yin give an $O(\log r)$-approximation algorithm for the problem, where $r = \max_{S \in \SetS} |S|$.  This was improved to a constant factor randomized approximation algorithm by Bansal, Gupta, and Krishnaswamy \cite{BansalGK10}.  They introduce a new linear programming relaxation for the problem and show how to use randomized rounding to achieve a performance guarantee of 485.\footnote{They observe in their paper that they did not try to optimize the constants in their analysis.}  In this paper, we show that by altering their algorithm using some concepts from $\alpha$-point scheduling (see Skutella \cite{Skutella06} for a survey), we are able to improve their algorithm and obtain a performance guarantee of about 28, which is an order of magnitude better.\footnote{Here we would like to point out that $28\in O(\sqrt{485})$.}

We now briefly review their algorithm and analysis, and then state the ideas we introduce to obtain an improvement.  Their algorithm begins with solving the following linear programming relaxation of the problem, where the variable $y_{S,t}$ for $t \in[n]$ (here and in the following the set $\{1,\ldots,n\}$ is denoted by $[n]$) and set $S \in \SetS$ indicates whether $C_S < t$ or not, and $x_{e,t}$ for $e \in U$ and $t \in [n]$ indicates whether element $e$ is assigned to the $t$th position of the ordering:
\begin{align*}
	\min &\sum_{t\in[n]}\sum_{S\in\SetS} (1-y_{S,t})\\
	\text{s.t.}\quad&\sum_{e\in U}x_{e,t}=1, && \text{for all~$t\in[n]$},\\
	&\sum_{t\in[n]}x_{e,t}=1, && \text{for all $e\in U$},\\
	&\sum_{e\in S\setminus A}\sum_{t'<t}x_{e,t'}  \geq (K(S)-|A|)\cdot y_{S,t}, && \text{for all $S\in \SetS,\,A\subseteq S,\,t\in[n]$},\\
	&x_{e,t},y_{S,t}\in[0,1], && \text{for all $e\in U,\,S\in\SetS,\,t\in [n].$}
\end{align*}
Bansal et al.\ observe that the exponentially many constraints can be separated in polynomial time such that the linear program can be solved efficiently. Let $x^*,y^*$ be an optimal solution. The algorithm proceeds in a sequence of $\lceil \log n \rceil$ stages.  In the $i$th stage, the algorithm of Bansal et al.\ computes a probability $p_{e,i} := \min\{1, 8 \sum_{t < 2^i} x^*_{e,t}\}$ for each element $e \in U$ by taking the amount that element $e$ is fractionally scheduled up to time~$2^i$ and boosting it by a factor of~$8$.  With probability $p_{e,i}$ it includes element $e$ in a set $O_i$.  If~$|O_i| > 16 \cdot 2^i$, the algorithm randomly chooses $16 \cdot 2^i$ elements from~$O_i$ and discards the remainder from~$O_i$.  For each $i$, the algorithm picks an arbitrary order for the elements in~$O_i$, then schedules the elements in the order $O_1, O_2, \ldots, O_{\lceil \log n \rceil}$.  Notice that it is possible that an element will appear in more than one $O_i$ and is scheduled more than once; one can compute an ordering that keeps only the first occurrence of each element.

The analysis of Bansal et al.\ works by looking at a time $t^*_S$ for each $S \in \SetS$, which is the smallest $t$ such that $y^*_{S,t} > 1/2$.  The analysis then shows that for any stage $i$ with $t^*_S \leq 2^i$, with probability at least $1-e^{-1}$ at least $K(S)$ elements have been marked in stage $i$ and no elements are discarded from $O_i$.  From this, the analysis infers that $E[C_S] \leq 64 \cdot \frac{e}{e-2} \cdot t^*_S$.  Since the value of the linear program is at least $\frac{1}{2} \sum_{S \in \SetS} t^*_S$, the paper derives that the expected value of the solution is at most $128 \cdot \frac{e}{e-2} \approx 484.4$ times the value of the linear program.

While we still use several ideas from their algorithm and analysis, we modify it in several key ways.  In particular, we discard the idea of stages, and we use the idea of a random $\alpha$-point for each element $e$; in particular, after modifying the solution $x^*$ to a solution $x$ in a way similar to theirs, we then randomly choose a value $\alpha_e \in [0,1]$ for each~$e \in U$.  Let $t_{e,\alpha_e}$ be the first time $t$ for which $\sum_{t'=1}^t x_{e,t'} \geq \alpha_e$.  We then schedule elements $e$ in the order of nondecreasing $t_{e,\alpha_e}$.  The improvements in analysis come from scrapping the stages (so we don't need to account for the possibility of $O_i$ being too large) and using $\alpha$-point scheduling; in particular, we introduce a parameter $\alpha$ and look for the last point in time $t_{S,\alpha}$ in which $y^*_{S,t}<\alpha$ (the Bansal et al.\ paper uses $\alpha=1/2$). Choosing $\alpha$ randomly gives our ultimate result.  We turn to the full analysis in the next section.

\section{The Algorithm and Analysis}

Let $x^*,y^*$ be an optimum solution to the linear program above. Let $\Q>0$ be a constant determined later. Construct a new solution $x$ from $x^*$ as follows: Initialize~$x:=\Q\cdot x^*$; for $t=1$ to $\lfloor n/2\rfloor$ set
\begin{align*}
x_{e,2t}:=x_{e,2t}+x_{e,t}~.
\end{align*}

\begin{lemma}\label{lem:}
For each $t\in[n]$
\begin{align} \label{eq:upperbnd}
\sum_{t'=1}^t\sum_{e\in U}x_{e,t'} \leq 2\cdot\Q\cdot t~.
\end{align}
Moreover, for each $e\in U$ and $t\leq\lfloor n/2\rfloor$
\begin{align}
\sum_{t'=t+1}^{2t}x_{e,t'}\geq \Q\sum_{t'=1}^tx^*_{e,t'}~,
\label{eq:geometrically-incr}
\end{align}
and for each~$t\in[n]$
\begin{align}
\sum_{t'=1}^{t}x_{e,t'}\geq \Q\sum_{t'=1}^tx^*_{e,t'}~.
\label{eq:geometrically-incr2}
\end{align}
\end{lemma}

\begin{proof}
We start by giving an alternative view on the definition of $x$ above. Notice that
\begin{align}
x_{e,t'} = \Q\sum_{t''\in I(t')} x^*_{e,t''}\qquad\text{with $I(t'):=\{t'': t'=2^i\cdot t''\text{~for some~}i\geq 0\}$.}	
\label{eq:alt_def_x}
\end{align}
That is, $I(t')$ is precisely the subset of indices $t''$ such that $x^*_{e,t''}$ contributes to $x_{e,t'}$.
For a fixed $t\in[n]$ and $t''\leq t$, let $J(t'')$ be the subset of all indices $t'\leq t$ such that~$x^*_{e,t''}$ contributes to $x_{e,t'}$, i.\,e., $J(t'')=\{t'\leq t: t''\in I(t')\}$. By definition of~$I(t')$ and $J(t'')$ we get
$\sum_{t'=1}^t|I(t')|=\sum_{t''=1}^t|J(t'')|$. Also notice that $|J(t'')|=1+\lfloor\log(t/t'')\rfloor$. Thus,
\begin{align*}
\frac1\Q\sum_{t'=1}^t\sum_{e\in U}x_{e,t'}	
&= \sum_{t'=1}^t\sum_{t''\in I(t')}\underbrace{\sum_{e\in U}x^*_{e,t''}}_{=1}
= \sum_{t'=1}^t|I(t')|
= \sum_{t''=1}^t|J(t'')|\\
& = t+\sum_{t''=1}^t \lfloor\log(t/t'')\rfloor
\leq t+\int_0^t\lfloor\log(t/\theta)\rfloor\,d\theta\\
&=t+\sum_{i=0}^\infty\int_{t/2^{i+1}}^{t/2^i}\lfloor\log(t/\theta)\rfloor\,d\theta
=t+\sum_{i=0}^\infty\frac{t}{2^{i+1}}\cdot i
=2t~.
\end{align*}
This concludes the proof of~\eqref{eq:upperbnd}.

In order to prove~\eqref{eq:geometrically-incr}, simply notice that for each $t''\in\{1,\dots,t\}$ there is~$t'\in\{t+1,\dots,2t\}$ such that $t''\in I(t')$; then \eqref{eq:geometrically-incr} follows from \eqref{eq:alt_def_x}. Finally, \eqref{eq:geometrically-incr2} also follows from \eqref{eq:alt_def_x} since $t'\in I(t')$ for all~$t'$.
\end{proof}

\paragraph{\bf Algorithm:} As discussed above, for each $e\in U$ we independently choose $\alpha_e\in[0,1]$ randomly and uniformly.  Let $t_{e,\alpha_e}$ denote the first point in time $t$ when $\sum_{t'=1}^{t}x_{e,t'}\geq\alpha_e$. In our final solution, we sequence the elements $e\in U$ in order of nondecreasing $t_{e,\alpha_e}$; ties are broken arbitrarily.

For $S\in\SetS$ and some fixed $\alpha\in(0,1)$, let $t_{S,\alpha}$ be the last point in time $t$ for which $y^*_{S,t} < \alpha$. We observe that the contribution of set $S$ to the objective function of the linear program is
\begin{align}
C_S^{LP}:=\sum_{t\in[n]} (1-y^*_{S,t}) = \int_0^1 t_{S,\alpha}\,d\alpha~,
\label{eq:C_SLP}	
\end{align}
since in time step $t$ it holds that $t_{S,\alpha} \geq t$ for all $\alpha \in [0,1]$ such that $\alpha > y^*_{S,t}$, or for $(1 - y^*_{S,t})$ of the possible $\alpha$.

We now bound the probability that we have fewer than $K(S)$ elements from $S$ with $t_{e,\alpha_e} \leq t_{S,\alpha}$ in terms of $\Q$ and $\alpha$.

\begin{lemma} \label{lem:probbnd}  Suppose $\Q\cdot\alpha\geq 1$.  Define $p$ such that
$$p := \exp\left(-\frac{1}{2}\cdot\left(1 - \frac{1}{\Q\cdot\alpha}\right)^2\cdot\Q\cdot\alpha\right) \leq 1~.$$  Then for integer $i \geq 0$,
$$\Pr\left[\bigl|\{e \in S: t_{e,\alpha_e} \leq 2^i\cdot t_{S,\alpha}\}\bigr| < K(S)\right] \leq p^{i+1}~.$$
\end{lemma}

\begin{proof}
Our analysis follows some of the analysis of Bansal et al.\ for a stage.  Let
$$A := \Biggl\{ e \in S: \sum_{t' \leq 2^i\cdot t_{S,\alpha}} x_{e,t'} \geq 1 \Biggr\}~.$$
Then observe that for any $e \in A$ it holds that $\Pr[t_{e,\alpha_e} \leq 2^i\cdot t_{S,\alpha}]=1$.  By the properties of the linear program,
$$\sum_{e \in S\setminus A} \sum_{t' \leq t_{S,\alpha}} x^*_{e,t'} \geq\bigl(K(S)-|A|\bigr)\cdot y^*_{S,1+t_{S,\alpha}} \geq \bigl(K(S)-|A|\bigr)\cdot\alpha~,$$
so that by \eqref{eq:geometrically-incr2}
$$\sum_{e \in S\setminus A} \sum_{t' \leq t_{S,\alpha}} x_{e,t'}  \geq\bigl(K(S)-|A|\bigr)\cdot\Q\cdot\alpha~.$$
More generally, it follows from induction on $i$ and \eqref{eq:geometrically-incr2} and \eqref{eq:geometrically-incr}, that
$$\sum_{e \in S\setminus A} \sum_{t' \leq 2^i\cdot t_{S,\alpha}} x_{e,t'}  \geq (i+1)\cdot\bigl(K(S)-|A|\bigr)\cdot\Q\cdot\alpha~.$$
For any $e \in S \setminus A$, let random variable $X_e$ be~$1$ if $t_{e,\alpha_e} \leq 2^i\cdot t_{S,\alpha}$ and~$0$ otherwise.  Note that $\Pr[X_e = 1] = \sum_{t' \leq 2^i\cdot t_{S,\alpha}} x_{e,t'}$.  Let $X := \sum_{e \in S \setminus A} X_e$ and $\mu := E[X]$, so that
$$\mu = E[X] = \sum_{e \in S\setminus A} \sum_{t' \leq 2^i\cdot t_{S,\alpha}} x_{e,t'}  \geq (i+1)\cdot\bigl(K(S)-|A|\bigr)\cdot\Q\cdot\alpha~.$$

Note that if $|A| \geq K(S)$, then $\Pr\left[|\{e \in S: t_{e,\alpha_e} \leq 2^i\cdot t_{S,\alpha} \}| < K(S)\right] = 0$, so we assume that $|A| < K(S)$. Then
\begin{align*}
\Pr&\left[|\{e \in S: t_{e,\alpha_e} \leq 2^i\cdot t_{S,\alpha} \}| < K(S)\right]\\
& = \Pr\left[|\{e \in S \setminus A: t_{e,\alpha_e} \leq 2^i\cdot t_{S,\alpha}\}| < K(S)-|A|\right]\\
& = \Pr\left[X < K(S)-|A|\right]\\
& \leq \Pr\left[X < \frac{\mu}{(i+1)\cdot\Q\cdot\alpha}\right]
= \Pr \left[X < \mu\cdot\left(1 - \left(1-\frac{1}{(i+1)\cdot\Q\cdot\alpha}\right)\right)\right]\\
& \leq \exp\left(-\frac{1}{2}\cdot\left(1 - \frac{1}{(i+1)\cdot\Q\cdot\alpha}\right)^2\cdot\mu\right)\\
& \leq \exp\left(-\frac{1}{2}\cdot\left(1 - \frac{1}{(i+1)\cdot\Q\cdot\alpha}\right)^2\cdot (i+1)\cdot\Q\cdot\alpha\right)\\
& \leq \exp\left(-\frac{1}{2}\cdot\left(1 - \frac{1}{\Q\cdot\alpha}\right)^2\cdot (i+1)\cdot\Q\cdot\alpha\right) = p^{i+1}
\end{align*}
where we use the Chernoff bound $\Pr[X < \mu\cdot(1 - \beta)] \leq \exp(-\frac12\cdot\beta^2\cdot\mu)$ (see, for example, Motwani and Raghavan \cite[Section 4.1]{MotwaniR95}), and the fact that $$-\left(1 - \frac{1}{(i+1)\cdot\Q\cdot\alpha}\right)^2 \leq -\left(1 - \frac{1}{\Q\cdot\alpha}\right)^2$$ for $i \geq 0$ and $\Q\cdot\alpha \geq 1$.
\end{proof}

Let $C_S$ be a random variable giving the position of the $K(S)$th element of $S$ in the ordering we construct, and let $C_S^{LP}$ be the contribution of set $S$ to the objective function as defined in~\eqref{eq:C_SLP}.  Then we can bound the cost of our schedule as follows, where $OPT_{LP} = \sum_{S \in \SetS} C_S^{LP}$ and $OPT$ is the cost of an optimal schedule.

\begin{lemma}\label{lem:fixed-alpha}
If $\Q$ and $\alpha$ are chosen such that $p<1/2$, then
$$\E\left[\sum_S C_S\right] \leq \frac{2\cdot\Q}{1-\alpha}\cdot \frac{1-p}{1-2p}\cdot OPT_{LP} + OPT~.$$
\end{lemma}

\begin{proof}
Let $t_S$ be the first point in time when $|\{e\in S: t_{e,\alpha_e} \leq t_{S}\}|\geq K(S)$. Then by Lemma \ref{lem:probbnd}, we know that the probability that $t_{S,\alpha} < t_S \leq 2\cdot t_{S,\alpha}$ is at most $p$, since the probability that $t_S > t_{S,\alpha}$ is at most $p$ by itself.  Similarly, the probability that $2\cdot t_{S,\alpha} < t_S \leq 4\cdot t_{S,\alpha}$ is at most $p^2$, the probability that $4\cdot t_{S,\alpha} < t_S \leq 8\cdot t_{S,\alpha}$ is at most $p^3$, and so on, so that
\begin{align}
	\E[t_S]&\leq t_{S,\alpha}+t_{S,\alpha}\sum_{i=0}^{\infty}2^i\cdot p^{i+1}
	=t_{S,\alpha}\cdot\left(1+\frac{p}{1-2p}\right) = t_{S,\alpha}\cdot \frac{1-p}{1-2p}~.
	\label{eq:expectation-t_S}
\end{align}
Note that for all $t \leq t_{S,\alpha}$ it holds that $1 - y^*_{S,t} > 1 - \alpha$, so that $C_S^{LP} \geq t_{S,\alpha}(1 - \alpha)$, or $t_{S,\alpha} \leq C_S^{LP}/(1 - \alpha).$  Thus
$$\E[t_S] \leq C_S^{LP}\cdot\frac1{1-\alpha}\cdot \frac{1-p}{1-2p}~.$$

Observe that $C_S \leq |\{ e \in U\setminus S: t_{e,\alpha_e} \leq t_S\}| + K(S)$.  Note that for any fixed element $e \notin S$ and time $t$, the probability that $t_{e,\alpha_e} \leq t$ is $\min\{1, \sum_{t' \leq t} x_{e,t'}\}$, so that
$$E\left[|\{e \in U \setminus S: t_{e,\alpha_e} \leq t\}|\right] \leq \sum_{e \in U} \sum_{t' \leq t} x_{e,t'} \leq 2\Q\cdot t$$
by \eqref{eq:upperbnd}.  Then we have that
\begin{align}
	\E[C_S]	\leq 2\Q\cdot \E[t_S]+K(S) \leq \frac{2\Q}{1-\alpha}\cdot \frac{1-p}{1-2p}\cdot C_S^{LP} + K(S)~,
	\label{eq:expC_S}
\end{align}
from which it follows that
\begin{align*}
	\E\left[\sum_S C_S\right] \leq \frac{2\Q}{1-\alpha}\cdot \frac{1-p}{1-2p}\cdot OPT_{LP} + OPT~,
\end{align*}
since in any solution $\sum_{S \in \SetS} K(S) \leq OPT$.
\end{proof}

We try to tune the various parameters to obtain the best possible performance guarantee.  If we set $\alpha:=1/2$ (as did Bansal et al.\ \cite{BansalGK10}) and $\Q:=10.05$, then $p=0.1995$, and thus we obtain
$$\frac{2\Q}{1 - \alpha} \cdot \frac{1-p}{1-2p} + 1 \approx 54.54~,$$
for a performance guarantee of about 55. However, we can do better if we choose~$\alpha$ (and~$\Q$) randomly.

\begin{theorem}
If we choose $\alpha$ independently at random from $(0,1)$ according to the density function $f(\alpha) = 2\cdot\alpha$ and set $\Q:=z/\alpha$ for some appropriately chosen constant $z$, then the algorithm has performance guarantee less than~$27.78$.
\end{theorem}

\begin{proof}
Notice that $\alpha\cdot\Q$ is equal to the fixed constant $z$ and $p=\exp\left(-\frac{1}{2}\cdot\left(1 - \frac{1}{z}\right)^2\cdot z\right)$ depends only on $z$ and is thus constant.

In the proof of Lemma~\ref{lem:fixed-alpha} we have obtained bounds on the expectations of~$t_S$ and~$C_S$ under the assumption that the values of~$\alpha$ and~$Q$ are fixed. We refer to these conditional expectations by $\E_{\alpha}$ such that
\begin{align*}
\E_{\alpha}[t_S]&\leq t_{S,\alpha}\cdot \frac{1-p}{1-2p} && \text{due to~\eqref{eq:expectation-t_S}, and}\\
\E_{\alpha}[C_S]&\leq 2\Q\cdot \E_{\alpha}[t_S]+K(S) && \text{due to~\eqref{eq:expC_S}.}
\end{align*}
Unconditioning together with~\eqref{eq:C_SLP} then yields
\begin{align*}
\E[C_S] &= \int_{0}^1 f(\alpha)\cdot E_{\alpha}[C_S]\,d\alpha\\
&\leq \int_{0}^1 2\alpha\cdot 2\Q\cdot t_{S,\alpha}\cdot\frac{1-p}{1-2p}\,d\alpha+K(S)\\
&=4z\cdot\frac{1-p}{1-2p}\int_{0}^1 t_{S,\alpha}\,d\alpha+K(S)\\
&=4z\cdot\frac{1-p}{1-2p}\cdot C_S^{LP} + K(S)~.
\end{align*}
Thus, we get
\begin{align*}
\E\left[\sum_{S\in\SetS} C_S\right] \leq 4z\cdot\frac{1-p}{1-2p}\cdot OPT_{LP} + OPT\leq \left(1+4z\cdot\frac{1-p}{1-2p}\right)\cdot OPT~.
\end{align*}
If we set $z:=5.03$, then $p\approx 0.1990$, and the performance guarantee is less than~$27.78$.
\end{proof}

\subsection*{Acknowledgements}
The first author was supported by the DFG Research Center \textsc{Matheon} "Mathematics for key technologies" in Berlin.
This work was carried out while the second author was on sabbatical at TU Berlin.
He wishes to acknowledge that he was supported in part by the Berlin Mathematical School, the Alexander von Humboldt Foundation, and NSF grant CCF-0830519.

\bibliographystyle{abbrv}
\bibliography{genminsumsc}

\end{document}